\documentclass[11pt,letterpaper]{article}
\usepackage[margin=1in]{geometry}
\pdfoutput=1

\usepackage[ruled]{algorithm2e} % For algorithms

\SetAlFnt{\small}
\SetAlCapFnt{\small}
\SetAlCapNameFnt{\small}
\SetAlCapHSkip{0pt}
\IncMargin{-\parindent}

\usepackage[margin=1in]{geometry}
\pdfoutput=1

\usepackage{times}
\usepackage{helvet}
\usepackage{courier}
\usepackage{amsthm, amssymb, amsmath}
\usepackage{graphicx}
\usepackage{verbatim}
\usepackage{empheq}
\usepackage{xifthen}
\usepackage{refcount}
\usepackage{cuted}
\usepackage{float}
\usepackage{caption}
\usepackage{subcaption}
\usepackage[export]{adjustbox}

\def\EE{{\mbox{E}}}

\def\eqpun{\hspace{0.5em}}

\def\union{\cup}
\def\reals{\mathbb{R}}
\def\cost{c}

\newtheorem{theorem}{Theorem}
\newtheorem{lemma}{Lemma}

\newtheorem{proposition}{Proposition}
\newtheorem{definition}{Definition}

\begin{document}

\title{Incentives for Truthful Evaluations}

\author{ Luca de Alfaro\footnote{The authors are listed in alphabetical order.}\\
luca@ucsc.edu, m.faella@unita.it \\
       {Computer Science Department}\\
       {University of California}\\
       {Santa Cruz, CA, 95064, USA}\\
\and
    Marco Faella\\
    m.faella@unina.it \\
    {Dept. of Electrical Engineering and Information Technologies} \\
    {University of Naples “Federico II”, Italy} \\
\and 
Vassilis Polychronopoulos \;\; Michael Shavlovsky\\
      \{vassilis, mshavlov\}@soe.ucsc.edu\\
       {Computer Science Department}\\
       {University of California}\\
       {Santa Cruz, CA, 95064, USA}\\
}
%\author{Luca de Alfaro}%\titlenote{In alphabetical order}}
%\affiliation{%
%  \institution{University of California, Santa Cruz}
%  \department{Computer Science Department}
%  \city{Santa Cruz}
%  \state{CA}
%  \postcode{95064}
%  \country{USA}}
%
%\author{Marco Faella}
%\affiliation{%
%  \institution{University of Naples ``Federico II''}
%  \department{Dept. of Electrical Engineering and Information Technologies}
%  \country{Italy}}
%
%\author{Vassilis Polychronopoulos}
%\affiliation{%
%  \institution{University of California, Santa Cruz}
%  \department{Computer Science Department}
%  \city{Santa Cruz}
%  \state{CA}
%  \postcode{95064}
%  \country{USA}}
%
%\author{Michael Shavlovsky}
%\affiliation{%
%  \institution{University of California, Santa Cruz}
%  \department{Computer Science Department}
%  \city{Santa Cruz}
%  \state{CA}
%  \postcode{95064}
%  \country{USA}}
%

\date{April, 2017}

\maketitle

\begin{abstract}
We consider crowdsourcing problems where the workers are asked to provide evaluations for items; the worker evaluations are then used to estimate the true quality of items.
Lacking an incentive scheme, workers have no motive in making effort in completing the evaluations, providing inaccurate answers instead.
We show that a simple approach of providing incentives by assessing randomly chosen workers is not scalable: to guarantee an incentive to be truthful the number of workers that the supervisor needs to assess grows linearly with total number of workers.
To address the scalability problem, we propose incentive schemes that are truthful and cheap: the truthful as the optimal worker behavior consists in providing accurate evaluations, and cheap because the truthfulness is achieved with little supervision cost.
We consider both discrete evaluation tasks, where an evaluation can be done either correctly, or incorrectly, with no degrees of approximation in between, and quantitative evaluation tasks, where evaluations are real numbers, and the error is measured as distance from the correct value.  
For both types of tasks, we develop hierarchical incentive schemes that can be effected with a small amount of supervised evaluations, and that scale to arbitrarily large crowd sizes: they have the property that the strength of the incentive does not weaken with increasing hierarchy depth.
We show that the proposed hierarchical schemes are robust: they provide incentives in heterogeneous environments where workers can have limited proficiencies, as long as there are enough proficient workers in the crowd.
Interestingly, we also show that for these schemes to work, the only requisite is that workers know their place in the hierarchy in advance.
\end{abstract}

%\sloppy

\section{Introduction}

Crowdsourcing allows access to large populations of human workers, and it can be an efficient and cheap solution for many applications. 
The very farming out of work to many independent workers, however, creates the problem of quality control.
In the absence of effective supervision or quality-control mechanisms, the workers may submit low quality work, or they may deliberately engage in straight-out vandalism. 
Workers can also collude with each other to game the system and collect rewards without performing the required work. 
In this paper, we describe supervisory schemes that provide an incentive towards high-quality work, and we show that the incentive is both cheap in terms of the required supervisor time and work overhead, and effective in making honest and accurate work the best strategy for workers.

We focus on crowdsourcing tasks which are \textit{verifiable}, that is, they have objective answers that a supervisor or another worker can check to conclude whether a worker is submitting quality work or not.
We further consider two types of verifiability: {\em binary,} and {\em quantitative.}
In binary verifiable tasks, the question of whether a worker submits quality work can be answered with either a Yes or a No.
In quantitatively verifiable tasks, the quality question can be answered only quantitatively as a measure of distance between the work submitted by a worker, and the work that was expected. 
Classification tasks are examples of binary verifiable tasks, as supervisors can check that the classification in discrete categories submitted by a worker matches expectations.
Grading tasks are examples where quantitative evaluation is natural; the quality of the human worker can be determined as a function of the distance of the worker's answer to the true grade.

We propose schemes that provide truthful incentives to the human workers at a low cost for the supervisors regardless of the size of the human worker population.

Using golden sets is a practice for quality control in crowdsourcing
\cite{venetis2012quality}.
Golden sets are sets of tasks for which the answer is known to the supervisor and are presented to workers with the goal of evaluating their performance;
such sets showed a positive impact on worker performance in crowdsourcing systems
\cite{harris2011you}.
Golden sets, however, can be difficult and costly to obtain \cite{oleson2011programmatic}, as they create an overhead for workers to perform extra tasks for quality control.
In applications that already require significant incentive to extract a small amount of items of work from workers, wasting work on golden set tasks is undesirable if there is an alternative.
Moreover, golden sets might be unavailable in advance.
For example, in peer grading of homework assignments, a golden set would need to be constructed for each homework assignment (unless the assignments are identical).
In addition, the golden set approach is problematic from a mechanism design point of view: knowing that the supervision is performed in this way, workers can infer the identity of golden set tasks by intersecting allocated tasks of different workers, and minimize effort by only being truthful in the golden tasks.
For class homework, where students can communicate with each other and some distinctive features may be easy to spot (`Did you also have to grade the Android app where the ball goes through the green wall?'), the golden set approach can be particularly weak.

The schemes that we propose rely on comparing the answers given by workers performing the same tasks with each other; in particular, they do not require golden sets of tasks for which the answer is known in advance. 

We first study a simple \textit{one-level} scheme and thereafter we propose a \textit{hierarchical} scheme.

In the simple \textit{one-level} scheme, workers perform tasks which are directly evaluated by the supervisor with some probability.
We study the conditions that ensure that workers maintain the incentive to provide truthful answers.

The one-level scheme does not scale to large crowds, as the supervisor needs to perform an amount of work that grows linearly with the number of workers. 
Thus, we introduce a {\em hierarchical\/} scheme, where the work of the supervisor is bounded even as the number of tasks and workers grows.
The scheme organizes workers in a hierarchy, where the supervisor is at the top, and the other workers are arranged in layers below.
Every worker in the hierarchy shares one common task with each worker below, so that it can verify part of the work performed by lower levels of the hierarchy. 
This hierarchical verification scheme entails no wasted work, and provides a truthful incentive to the workers regardless of their level in the hierarchy.
The scheme is based on one, uniform, category of workers: we do not need to split workers into ``regular'' workers and meta-reviewers.
As the worker population increases, the hierarchy becomes deeper, but the amount of work that the supervisor needs to do remains constant, and so does the incentive towards correct behavior.
We show that the only information about the hierarchy that needs to be communicated beforehand to the workers is their level in the hierarchy itself. 
We provide matching upper and lower bounds for the amount of information that needs to be communicated beforehand to workers in the hierarchy to maintain a truthful incentive, showing that a logarithmic amount of information in the number of workers is both necessary and sufficient.

We study the practical aspects of the implementation of the hierarchy.
Many crowdsourcing tasks benefit from redundancy, that is, from assigning the same task to more than one worker. 
For instance, by assigning the same item to multiple graders, it is possible to reconstruct a higher-accuracy grade for the item than would be available from one grader alone \cite{piech_tuned_2013}
We show that in redundant tasks in which there is no control over task allocation to workers, the problem of creating an optimal hierarchy is NP-hard. 
We present fast approximation algorithms that are optimal within constant factors.
If the supervisor can control the allocation of tasks to workers, as in many real applications, we show that constructing the hierarchy is an easy problem. 

We develop our results first in the case of binary verifiable tasks.
These are common in classification tasks: spam or not, correct answer or not, etc.
We consider a model where workers need to make an ``effort'' of $f(e)$ in order to ensure that their error probability is lower than $e$. 
We obtain a tight lower bound for the mistake penalties necessary to ensure that the correctness incentive propagates to all levels of the hierarchy. We show that the truthful incentive holds even when the supervisor occasionally makes mistakes, and in populations of workers with diverse proficiency, where workers can have limited proficiency, provided that there are enough proficient workers in the crowd.

We then show how the results on binary verifiable tasks transfer to the case of quantitative tasks.
In quantitative tasks, the notion of a task performed correctly is replaced by the notion of variance in the quantitative outcome of the task.
The effort function relates the effort (or cost) to the worker to the variance in the worker's evaluation.
In the model, increased worker effort produces higher expected precision of the worker's answers, and is similar to other models proposed in the literature \cite{cai2014optimum}.
We show that we can shape the incentives to ensure that it is optimal for all players to put sufficient effort to ensure their variance is below a given threshold, independently of the worker position in the hierarchy.
In other words, hierarchical distance from the supervisor does not entail loss of precision in the tasks performed.
This enables the scheme to scale to arbitrarily large crowds, while keeping the work of the supervisor bounded and the precision constant.

The proposed schemes are thus applicable to a multitude of crowdsourcing applications, from conventional classification tasks using generic crowds in crowdsourcing marketplaces to peer grading in Massive Open Online Courses with an arbitrarily large population of students.

\section{Related Work}

Providing incentives to human agents to return truthful responses is one of the central challenges for crowdsourcing algorithms and applications \cite{ghosh2013gamechapter}.

Prediction markets are models with a goal of obtaining predictions about events of interest from experts.
After experts provide predictions, a system assigns a reward based on a scoring rule to every expert. Proper scoring rules ensure that the highest reward is achieved by reporting the true probability distribution
\cite{winkler1968good,johnson1990efficiency,clemen2002incentive}.
An assumption of the scoring rules is that the future outcome must be observable.
This assumption prevents crowdsourcing systems to scale to large crowds as obtaining the correct answer for each event or task is prohibitively expensive.

The model presented in \cite{carvalho2013} relaxes this assumption.
The proposed scoring rule evaluates experts by comparing them to each other.
The model assigns a higher score for an expert if her predictions are in agreement with predictions of other experts.
Work \cite{carvalho2013} belongs to the class of peer prediction methods.
Peer prediction methods is wide class of models for providing incentives 
\cite{serum,peerprediction,faltings2012eliciting,witkowski2012robust,witkowski2013dwelling,dasgupta2013crowdsourced,jurca2009mechanisms,radanovic2013robust,riley2014minimum,zhang2014elicitability,waggoner2014output,kamble2015truth,kong2016putting} .
Such methods elicit truthful answers by analyzing the consensus between workers in one form or another.
Peer prediction methods encourage cooperation between workers and, as a result,
promote uninformative equilibria.
The study in \cite{jurca2005enforcing} shows that for the scoring rules proposed in the peer-prediction method \cite{peerprediction}, a strategy that always outputs ``good'' or ``bad'' answer is a Nash equilibrium  with a higher payoff than the truthful strategy.
Works by \cite{jurca2009mechanisms,waggoner2014output} show that
the existence of such equilibria is inevitable.
In contrast, hierarchical incentive schemes we propose make the truthful strategy the only Nash equilibrium.

The model described in \cite{jurca2009mechanisms} considers a scenario of rational buyers who report on the quality of products of different types.
In the developed payment mechanism the strategy of honest reporting is the only Nash equilibrium.
However, the model requires that the prior distribution over product types and condition distributions of qualities is the common knowledge.
This requirement is a strong assumption.

The Bayesian Truth Serum scoring method proposed in \cite{serum} elicits truthful subjective answers on multiple choice questions.
The author shows that the truthful reporting is a Nash equilibrium with the highest payoff.
The model is different from other approaches in that besides the answers, 
workers need to provide predictions on the final distribution of answers.
Workers receive a high score if their answer is ``surprisingly'' common - the actual
percentage of their answer is larger than the predicted fraction.
Similarly, incentive mechanisms in  \cite{witkowski2012robust,witkowski2013dwelling,radanovic2013robust,radanovic2014incentives,riley2014minimum} require workers provide belief reports along with answers on tasks.
Truthful mechanisms in \cite{peerprediction,zhang2014elicitability,kong2016putting} requires knowledge about the distribution from which answers are drawn.
Our mechanisms do not rely on worker's beliefs on other workers' responses nor require knowledge about the global answer distribution.

The work in \cite{Alon:2011} studies the problem of incentives for truthfulness in a setting where persons vote other persons for a position. The analysis derives a randomized approximation technique to obtain the higher voted persons. 
The technique is strategyproof, that is, voters (which are also candidates) cannot game the system for their own benefit. 
The setting of this analysis is significantly different from ours, as the limiting assumption is that the sets of voters and votees are identical. 
Also, the study focuses on obtaining the top-$k$ voted items, while in our setting we do not necessarily rank items.
%Another $k$-selection method that provides truthful incentives is proposed in \cite{kurokawaimpartial}.

The PeerRank method proposed in \cite{Walsh2014} obtains the final grades of students using a fixed point equation similar to the PageRank method. 
However, while it encourages precision, it does not provide a strategyproof method for the scenario that students collude to game the system without making the effort to grade truthfully.

Authors of \cite{oleson2011programmatic} propose an automated process to generate golden tasks for quality assurance in crowdsourcing.
An initial set of golden tasks is used to bootstrap a larger set of golden tasks.
A tasks is chosen if it has several matching answers by the reliable workers, that is, workers who provided
correct answers to the original golden tasks.
The chosen tasks are then used to create new golden tasks by injecting common errors.
The step of error injecting is to ensure that common error types are present in then new golden set.
Note that the process of detecting common error types is manual.
The authors report on decreasing amount of manual work to manage large crowds.
The main difference with our work is that we provide theoretical guarantees that the proposed incentive mechanisms require constant amount of work by the supervisor for arbitrary large crowds.
Moreover, golden sets are not suitable for all applications.
For example, in peergrading of homework assignments the total set of homework submissions cannot be obtained
before the homework is posted. Also, information on the competence of workers from previous homeworks or classrooms cannot be used reliably for newer homeworks or in different classrooms with different material.
The incentive schemes we propose do not require golden sets.

%Work by \cite{sheng2008get}
%proposes to improve quality of data annotation quality by increasing redundancy on tasks that have high uncertainty
%about the quality of the current answers.
%This approach is susceptible to colluding and does not rule out uninformative equilibria.

Employees in organizations and firms are frequently organized into hierarchies.
Economists study incentives in hierarchical organizations, the influence of
hierarchical structure on firms sizes and the loss of control within hierarchies.
Previous studies on hierarchies relevant to ours are found in \cite{williamson1967hierarchical, calvo1978supervision, calvo1979hierarchy, qian1994incentives}

However, our work and the work by organizational economists is not directly comparable due to different models.
Our model is designed to reflect the nature of evaluation tasks.
Models by \cite{williamson1967hierarchical, calvo1978supervision, calvo1979hierarchy, qian1994incentives} are designed to reflect economic aspects of firms.
Our work is not directly comparable to the work of organizational economist, as the model we describe are
applicable to peer grading setting rather than corporate hierarchy setting.
For example, the model described in \cite{williamson1967hierarchical} assumes that workers on the bottom layer do the production work, 
while all other workers (managers) do the coordination and supervision work.
Subordinate workers satisfy requests by their superiors with a discount factor that is within $(0, 1)$ range.
The smaller the discount factor, the smaller the contributions by the workers to the firm's revenue.
In contrast, in our models, all workers perform evaluation work.
Workers are evaluated based on the comparison of their answers to the answers of their supervisors.
Our models admit that worker can make mistakes or have bounded proficiencies.
The work in \cite{williamson1967hierarchical} shows that there is a limit on the size of hierarchy due to loss of control.
In contrast,  the hierarchies of workers that we propose have the property that the incentive to do truthful evaluations does not deteriorate with the hierarchy depth.

\section{Crowdsourcing Models}

We consider two crowdsourcing models: the binary-verifiable model and the quantitative model.
In the binary-verifiable model, tasks have a property that a proposed solution can be verified by comparing with the correct solution.
%For example, classifying images from span/not spam categories is a binary-verifiable tasks.
For example, a task of classifying items from a discrete set categories is binary-verifiable.
A task of grading essays, in contrast, is not binary-verifiable: the correctness of an essay cannot be established based on a comparison with one designated essay.
However, the quality of an essay can be expressed via a numerical grade.
In the quantitative setting, the solution of a task is a real valued number.
We make these settings precise via the following models.

Let $U$ and $I$ be the set of workers and tasks respectively.
Every worker $u$ performs a subset of tasks from $I$.
We construct a bipartite graph $G=(U \union I, E\}$ with
tasks and workers as nodes.
For a task $i \in I$ and a worker $u \in U$, the edge $(i, u)$ belongs to the set of edges $E$ 
iff the worker $u$ was assigned the task $i\in I$.
We denote the set of tasks assigned to a worker $u$ as $\partial u$,
the set of workers assigned with a task $i$ as $\partial i$.

\subsection{The Binary-verifiable Model}

%We make the binary-verifiable setting precise via the following model.
In the binary-verifiable model, each task $i\in I$ has a solution from a set $A$.
A worker performs task $i$ by choosing a solution from set $A$.
Every tasks $i\in I$ has a correct solution $s_i\in A$.
To find the correct solution the worker needs to make effort.

\smallskip

\noindent{\bf Effort Function.}
Let $a_i$ be a solution proposed by a worker on a task $i\in I$, and
let $e$ be a probability that $a_i$ is a wrong solution, i.e.
$e = Pr(a_i \neq s_i)$.
The effort is defined by a function $f:(0, 1]\to [0, +\infty)$, 
such worker needs to pay cost $f(e)$ to have error probability at most $e$.
We require function $f$ to be monotonically decreasing (larger error bounds cost less), and to be differentiable and strictly convex.
Requiring that $f$ is convex does not entail any loss in generality.
For $x, y \in (0, 1]$ and $0 < \alpha < 1$, if we had 
$f((1-\alpha) x + \alpha y) > (1-\alpha) f(x) + \alpha f(y)$, 
contradicting convexity, then the worker would obtain a lower cost simply by paying $f(x)$ a fraction $1-\alpha$ of the time, and $f(y)$ a fraction $\alpha$ of the time, obtaining overall error probability equal to $(1-\alpha) x + \alpha y$ at a cost lower than $f((1-\alpha) x + \alpha y)$.
Strict convexity of $f$ entails that, the closer $x$ goes to 0, the more difficult it gets to reduce the error by the same amount.

\smallskip

\noindent{\bf Strategies.}
A worker's strategy is the choice of error probability $e$ and corresponding effort $f(e)$.
For a specified error threshold $\varepsilon \in (0,1]$, we call a strategy with error probability $e$ truthful iff $e < \varepsilon$.

\smallskip

\noindent{\bf Supervision.}
We assume that there is a supervisor that can verify whether tasks are done correctly or not. 
Let worker $u$ provides a solution $a_u$ on a task $i$, 
while the supervisor provides a solution $a$ the same task of the worker.
The supervisor assigns loss $l(a_u, a)$ to the worker defined by $l(a_u, a) = 0$ if $a_u = a$, and $l(a_u, a) = C$ if $a_u \neq a$, for a fixed punishment cost $C > 0$. 
Note that, when the supervisor verifies a task $i$, the supervisor can verify all the workers that also performed $i$, that is, all the workers in $\partial i$.

\subsection{The Quantitative Model}

In the quantitative model, the workers are asked to evaluate items from a set $I$, associating to each item $i \in I$ a value $q_i \in \reals$. 

\smallskip 

\noindent{\bf Effort Function.}
In order to produce a precise measurement of the quality of an item,
a worker needs to pay a price defined by an effort function.
To produce a measurement of the quality of an item to within variance $v$, a worker needs to pay a price $f(v)$, where $f$ is a non-negative, monotonically decreasing, strictly convex function defined on the set $\reals^{+}$ of strictly positive variances.
Again, the hypotheses that $f$ be strictly convex and monotonically decreasing are not restrictive.

\smallskip 

\noindent{\bf Strategies.}
A strategy of a worker consists in choosing a precision (variance) $v$ and corresponding effort $f(v)$.
Similarly to the binary-verifiable model, let $\varepsilon \in \reals$ be a variance threshold.
We call a strategy $v$ truthful if $v < \varepsilon$.

\smallskip

\noindent{\bf Supervision.}
In order to produce an incentive towards precise work,
workers can be evaluated by the supervisor, or by a worker in a higher level.
If the worker produces estimate $x$, while the supervisor or upper-level worker produces estimate $y$, the worker is penalized using the loss function
\begin{align}
    \label{loss-quant}
  \ell(x, y) = \cost (x-y)^2 \eqpun ,
\end{align}
where $\cost > 0$ is a penalty constant.

\smallskip

In the following sections we will propose one level and hierarchical supervision schemes that provide incentive to workers so that they play with the truthful strategy.

\section{One Level Supervised Schemes}

In this section we study ``one-level'', or flat, supervision schemes where 
workers are directly verified by the supervisor.
To verify a worker $u\in U$, the supervisor examines a task $i\in \partial u$ assigned to the worker. 
The supervisor then imposes a loss $l$ to the worker depending on their solution.
To provide an incentive to workers $U$ to play with a truthful strategy,
the  supervisor chooses a subset of tasks to examine.
We use $p$ to denote the probability that a randomly chosen worker 
has a task that belongs to the subset.
By selecting first a random subset of $m$ workers, and then picking an item for each worker, the supervisor can ensure a probability at least $p = m/|U|$ of verifying a worker. 
The higher the probability $p$, the higher the influence of the supervisor on all workers.
We show that the number of items the supervisor needs to evaluate grows linearly with the numbers of items.
The result of this section is very similar to the one in the work by \cite{gao2016incentivizing}.
We consider a setting where the supervisor provides precise answers without making mistakes.
Such an assumption simplifies the proofs of theorems in this section, without limiting the generality of
the results, as we consider the most favorable setting for the instructor to provide incentives.
Still, the number of items to evaluate grows linearly with the total number of workers.
We will relax this assumption in the future sections and consider settings when the supervisor makes mistakes too.
%%We show that the supervisor of a flat incentive scheme has to evaluate items items that grows linearly  require  to guarantee incentives to be truthful 

\smallskip

\noindent {\bf Binary-verifiable model.}
Theorem~\ref{th-lower-bound-p} establishes a lower bound on $p$ for the binary-verifiable model so that workers have an incentive to be truthful.
Note that $-f'(x) > 0 $ as the effort function is monotonically decreasing.

\begin{theorem}
\label{th-lower-bound-p}
If every worker is assigned $k$ tasks, the penalty cost equals $C$, and the probability $p$ of being verified by the supervisor satisfies the
following inequality:
\begin{align}
    \label{ineq-flat-p}
    p > \frac{(-f'(\varepsilon))k}{C} ,
\end{align}
then workers minimize their loss by playing with a truthful strategy.
\end{theorem}
\begin{proof}
    The expected loss $L$ of a worker $u\in U$ consists of two components: the effort to perform $k$ tasks, and the expected penalty due to the supervisor when the worker provides a wrong solution%; $\loss = kf(e) + epC$.
    \begin{align}
        \label{eq-flat-loss-binar}
        L(e) = kf(e) + epC \eqpun .
    \end{align}
    The expected loss $L(e)$ is a convex function of $e$ as a sum of a decreasing strict convex and an increasing linear function, where the decreasing convex function is bounded from below by 0.
    Thus it has a global minimum $e^*$ that satisfies the equality
    \begin{align}
        \label{eq-flat-p-1}
        kf'(e^*) + pC = 0 \Rightarrow p = \frac{-f'(e^*)k}{C} \eqpun .
    \end{align}
    We substitute the probability $p$ in inequality (\ref{ineq-flat-p}) with the expression of (\ref{eq-flat-p-1}) and obtain
    \begin{align}
        \label{eq-flat-p-2}
        \frac{-f'(e^*)k}{C} &> \frac{-f'(\varepsilon)k}{C} \Rightarrow %\\[1ex]
        f'(e^*) < f'(\varepsilon) \eqpun .
    \end{align}
    Because function $f'(e)$ is increasing as the derivative of a strictly convex function \cite{rockafellar}, it follows from the inequality
    $f'(e) < f'(\varepsilon)$ that $e < \varepsilon$.
    Therefore, to minimize the loss, a worker chooses to play with an error probability $e^*$ such that $e^* < \varepsilon$,
    thus being truthful.
\end{proof}

\noindent {\bf Quantitative model.}
For the quantitative model, the following theorem establishes a similar lower bound on the probability $p$ of being verified by the supervisor.
\begin{theorem}
\label{th-lower-bound-p-quant}
If every worker is assigned $k$ tasks, the penalty constant of loss (\ref{loss-quant}) is $\cost$ , and the probability $p$ of being verified by the supervisor satisfies the following inequality
\begin{align}
    \label{ineq-flat-p-quant}
    p > \frac{(-f'(\varepsilon))k}{\cost} ,
\end{align}
then workers minimize their loss by playing with a truthful strategy.
\end{theorem}
We omit the proof as it follows similar steps as in the poof of Theorem~\ref{th-lower-bound-p}.

For both models, the number of workers $m = p|U|$ that the supervisor needs to examine grows linearly with the total number of workers. 
This limits the applicability of the flat approach to relatively small task sets and worker crowds. 
In the following section, we will develop a hierarchical setting that overcomes this limitation.

\section{Hierarchical Supervised Schemes}
\label{sec-hierarchical}

In this section we develop hierarchical schemes that require
a fixed amount of work by the supervisor to provide an incentive to workers for doing diligent work, regardless of the total number of workers.
We first consider the case without redundancy when no task is guaranteed multiple workers; the case with redundancy, when every task is guaranteed multiple workers, is studied in Section~\ref{sec-redundancy} later. 
We develop hierarchical incentive schemes for both the binary-verifiable model and the quantitative model.
In these schemes all workers perform tasks; there are no special meta-review tasks. 
The tasks are assigned so that each worker shares at least one task with a worker one level above in the hierarchy. 
By comparing the answers of the workers on these shared tasks, the workers at upper levels effectively check the work of workers at lower hierarchical levels. 
We note that the workers do not know which tasks they share with other workers; all they need to know, as we will show, is their level in the hierarchy. 
We will show that an incentive to be truthful does not deteriorate as the depth of a hierarchy grows.

The scheme organizes workers into a {\em supervision tree} (see Definition~\ref{def-tree}).
The internal nodes of the supervision tree represent workers; the leaves
represent tasks.
A parent node and a child node share one task;
this shared item is used to evaluate the quality of the child node's review work.
At the root of the tree is the supervisor
who is truthful, that is, she has a small probability of mistakes.

\begin{definition}
  \label{def-tree}
  A supervision tree of depth $L$ is a tree with tasks as leaves, workers as internal nodes, and the supervisor as root.
  The nodes are grouped into levels $l = 0, \ldots,
  L-1$, according to their depth; the leaves are the nodes at level
  $L-1$ (and are thus all at the same depth).
  In the tree, workers at level $L-2$ perform the tasks they are connected to.
  Every node at level $0 \leq l < L - 2$ performs exactly
  one task in common with each of its children.
\end{definition}

To construct a supervision tree of branching factor at most $k$, we proceed as follows.
We place the tasks as leaves and the above level of workers with at most $k$ tasks per worker.
Once level $l$ is built, we build level $l-1$ by enforcing a branching
factor of at most $k$.
For each node $x$ at level $l$, let $y_1, \ldots, y_n$ be its children.
For each child $y_1, \ldots, y_n$, we pick at random a task $s_i$
performed by $y_i$, and we assign node $x$ with the task of examining the set
$\{s_1, \ldots, s_n\}$ of tasks.
At the root of the tree, we place the supervisor, following the same
method for assigning tasks, that is,  we assign the supervisor with doing one task from each of his or her children nodes, picked at random.
Figure~\ref{fig-review-tree} illustrates a supervision tree with branching
factor 2 and depth 3.

\begin{figure}[!ht]
\centering
\includegraphics[width=0.4\linewidth]{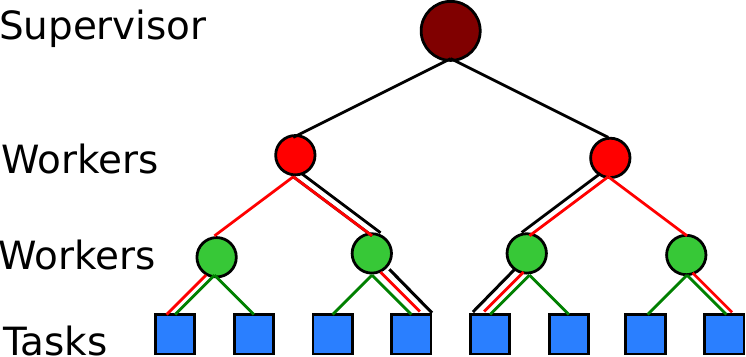}
\caption{An example of a supervision tree with branching factor 2. The
  process starts bottom up. Each worker is assigned 2 tasks.
For each depth-2 worker, a depth-1 worker is assigned 
one task in common with worker at depth-2 (red edges).
The evaluation of the depth-2 worker will depend on the depth-1 worker.
Similarly, the supervisor evaluates a depth-1 worker by reviewing one
of the two tasks that the depth-1 worker has done(black edges).}
\label{fig-review-tree}
\end{figure}

In the following two subsections we study hierarchical schemes with binary-verifiable and quantitative tasks respectively.

\subsection{The Binary-verifiable Model}
We consider the case of a homogeneous worker population first, where workers have the same effort function.
We find a tight bound (\ref{ineq-hierarch-c}) on penalty cost $C$ that ensures that a hierarchical scheme provides incentives.
We then extend our results to the case of heterogeneous worker population, where worker effort functions are different, modeling a more realistic setting.
The bound (\ref{ineq-hierarch-c}) will be crucial in distinguishing between workers who are proficient and workers with limited proficiency.
We will show that if a population of workers is proficient on average, then a hierarchical scheme provides
an incentive to be truthful.

\smallskip

\noindent{\bf Homogeneous worker population.} 
We consider a setting where all workers have the same effort function $f$.
We will show that the proposed hierarchical scheme provides incentives for worker to be truthful.
First, to prove our main results, we formulate Lemma~\ref{lemma-binar} that
considers a worker and their superior in a supervision tree.
It computes the expected loss of the worker, and it provides an upper bound on the worker's error probability.

\begin{lemma}
    \label{lemma-binar}
    Let workers $u, w \in U$ have error probabilities $e_u, e_w$ respectively,
    and let worker $w$ be the parent of worker $u$ in a supervision tree
    with branching factor $k$ and penalty cost $C>0$.
    If worker $u$ has effort function $f$ then the expected loss $L(e_u, e_w)$
    of worker $u$ under the supervision of worker $w$ is 
    \begin{align}
        \label{eq-expected-bin-loss}
        L(e_u, e_w) = kf(e_u) + e_u(1-e_w)C + (1-e_u)e_wC + e_ue_wD \eqpun ,
    \end{align}
    where $D$ is a constant from the $[0, C]$ interval.
    Moreover, if there is $\sigma, \varepsilon \in (0, 1/2)$ such that
    $e_w < \varepsilon$ and 
    \begin{align}
        \label{lemma-ineq-C}
        C \ge \frac{f'(\sigma)k}{2\varepsilon - 1} \eqpun ,
    \end{align}
    then every $e^*_u\in \underset{e_u}{\arg\min} L(e_u, e_w)$ satisfies inequality
    $e^*_u < \sigma$ .
  %  \begin{align*}
  %      e^*_u < \sigma \eqpun .
  %  \end{align*}
\end{lemma}
The proof of Lemma~\ref{lemma-binar} is provided in Appendix~\ref{sec-lemma-proof}.
The following theorem shows that in the proposed hierarchical scheme, we can choose the penalty $C$ so that all
Nash equilibria have the property that every worker plays truthfully.

\begin{theorem}
    \label{th-supervised-tree}
    Let workers be organized into a supervision tree with branching factor $k$.
    If workers are rational and have the  same effort function $f$, and penalty $C$ satisfies the following
    inequality
    \begin{align}
        \label{ineq-hierarch-c}
        C \ge \frac{f'(\varepsilon)k}{2\varepsilon - 1} \eqpun,
    \end{align}
    for $\varepsilon \in (0, 1/2)$,
    then all Nash equilibria have the property that every worker in an equilibrium plays truthfully.

\end{theorem}
\begin{proof}
    The strategic choice of a player $u$ depends only on the players placed above $u$.
    Thus the proof is by induction on the depth $l = 0, 1, \dots, L-1$  of the tree.
    The inductive hypothesis is that the best response for players at depth
    up to $l$ is to play with a truthful strategy, i.e. with error probability less that $\varepsilon$.
    At depth $0$, the result holds trivially, as the supervisor plays a fixed truthful strategy.

    To prove the inductive step,
    let us consider a worker $u$ at level $l+1$ and the worker's superior
    $w$ at level $l$, with error probabilities $e_u$ and $e_w$ respectively.
    The loss of worker $u$ is solely depends on the superior $w$.
    By the inductive hypothesis, the superior $w$ plays with a truthful strategy and therefore provides an incorrect solution 
    to the common task with probability $e_w < \varepsilon$.
    The conditions of Lemma~\ref{lemma-binar} for worker $u$ with superior $w$ are satisfied: 
    superior $w$ has error probability $e_w < \varepsilon$, and inequality (\ref{lemma-ineq-C}) holds for $\sigma=\varepsilon$.
    Therefore, it follows from Lemma~\ref{lemma-binar} that the best response of worker $u$ is to play
    with error probability $e^*_u$ that is less than $\varepsilon$, i.e. by choosing a truthful strategy.
\end{proof}

We can show that the inequality (\ref{ineq-hierarch-c})  is tight: 
if cost $C$ does not satisfy inequality (\ref{ineq-hierarch-c}), then there exists an effort function, a supervision tree, and a level $l$ such that a rational worker chooses a strategy with error probability greater than $\varepsilon$, that is, not a truthful strategy.

\begin{theorem}
    \label{th-counter-example}
    The bound proved in Theorem~\ref{th-supervised-tree} is tight.
    Precisely, if the cost $C$ does not satisfy inequality (\ref{ineq-hierarch-c}), then there exists an effort function, a supervision tree, and a level $l$ such that a rational worker chooses a strategy with error probability greater than $\varepsilon$.
\end{theorem}
%
%    We examine a setting where the opposite of inequality (\ref{ineq-hierarch-c}) holds.
%    Assume that workers are organized into a supervision tree with
%    branching factor $k$, and that solution set $A$ consists of 2 elements only.
%    Let penalty cost $C$ satisfy inequality
%    \begin{align*}
%        C < \frac{f'(\varepsilon)k}{2\varepsilon - 1}
%    \end{align*}
%    for $\varepsilon \in (0, 1/4)$ and effort function $f(x) = -\ln(x)$. 
%    Then we can show that there exists depth $l$ such that any tree with depth greater
%    or equal $l$ does not provide an incentive to be truthful:
%    a rational worker on level $l$ chooses error probability $e$ such that
%    $e > \varepsilon$.

A detailed proof of Theorem~\ref{th-counter-example} is provided in Appendix~\ref{th-counter-example-proof}.
\iffalse
    \begin{proof} (Sketch)
        We examine a setting where the opposite of inequality (\ref{ineq-hierarch-c}) holds.
        Assume that workers are organized into a supervision tree with
        branching factor $k$, and that solution set $A$ consists of 2 elements only.
        Let penalty cost $C$ satisfy the opposite of inequality (\ref{ineq-hierarch-c})
        for $\varepsilon \in (0, 1/4)$ and effort function $f(x) = -\ln(x)$. 
        Then we can show that there exists depth $l$ and a tree with depth greater
        or equal $l$ that does not provide an incentive to be truthful:
        a rational worker on level $l$ chooses error probability $e$ such that
        $e > \varepsilon$.
        A detailed proof is provided in Appendix~\ref{th-counter-example-proof}.
    \end{proof}
\fi

\smallskip

\noindent{\bf Heterogeneous worker population.}
We study the scenario when workers have different effort functions.
A worker in a supervision tree knows their own effort function;
however, they do not know effort functions of other workers.
Thus, strategic interactions between workers can be formulated as a game with incomplete information about effort functions.
We introduce the definition of a {\em proficient} worker and we show that if workers 
are proficient on average, then all {\em interim} Bayes-Nash equilibria have 
the property that every proficient worker plays truthfully.

Proficiency of a worker $u$ is determined by their effort function $f_u$.
A more proficient worker spends less effort to achieve the same error probability,
compared to a less proficient worker.
The fact that the lower bound (\ref{ineq-hierarch-c}) is tight,
allows us distinguish between a proficient worker and a worker with limited proficiency.

\begin{definition}
    \label{def-prof-worker}
    Given $\varepsilon \in (0, 1/2)$ and penalty $C>0$,
    worker $u\in U$ with effort function $f_u$ is {\em proficient} if the solution $\sigma_u$ of equation
    \begin{align}
        \label{eq-prof-level}
        C = \frac{f'_u(\sigma_u)k}{2\varepsilon - 1} \eqpun
    \end{align}
    satisfies inequality $\sigma_u \le \varepsilon$.
\end{definition}

\begin{comment}
To simplify technical details, we impose conditions on effort functions $f_u, u\in U$ to guarantee that there
exists a unique solution $\sigma_u$.
\end{comment}
First, we naturally assume that the more one tries to approach a zero probability of mistake, the more effort one needs to make, while having precisely zero probability would require infinite effort. This assumption translates into the following conditions: $\lim_{e\to +0} f_u(e) = +\infty$ and  $\lim_{e \to +0}f'_u(e)= -\infty$.
Secondly, random guessing requires no effort: $\lim_{e\to 1/2}f_u(e) = 0$ and $\lim_{e\to 1/2}f'_u(e) = 0$.
The two natural assumptions guarantee that the derivative of $f_u$ lies in the $(-\infty, 0)$ range.
Moreover, the differentiablity and the convexity of $f_u$ implies the continuity of $f_u'$ \cite{rockafellar},
therefore, the continuity of $f_u'$ and the strict convexity of $f_u$ guarantees that the derivative takes any value in that range only at one point, that is, there are no two points with the same derivative.
Thus, equation (\ref{eq-prof-level}) always has a unique solution.

Definition~\ref{def-prof-worker} is justified by the following observation.
Based on Lemma~\ref{lemma-binar}, value $\sigma_u$ is the upper bound on the best response $e^*_u$ of worker $u$ under a supervision of a truthful worker.
Therefore, error probability $e^*_u$ of a {\em proficient} worker satisfies inequality $e^*_u < \sigma_u \le \varepsilon$.
This means that a {\em proficient} worker is truthful under a truthful superior.
On the other hand, if a worker has limited proficiency, i.e. $\sigma_u > \varepsilon$, then, as we will show,
the opposite of inequality (\ref{ineq-hierarch-c}) holds;
and in this case, according to Theorem~\ref{th-counter-example}, a worker with that effort function in a supervision tree
is not guaranteed to be truthful.
We obtain the opposite of inequality (\ref{ineq-hierarch-c}) from inequality $\sigma_u > \varepsilon$ by using
the fact that $f'_u$ is an increasing function as a derivative of a strictly convex function and the assumption that
$\varepsilon < 1/2$:
\begin{align*}
    \sigma_u > \varepsilon \Rightarrow f'_u(\sigma_u) > f'_u(\varepsilon) \Rightarrow \frac{f'_u(\sigma_u)k}{2\varepsilon - 1} < \frac{f'_u(\varepsilon)k}{2\varepsilon - 1} \Rightarrow C < \frac{f'_u(\varepsilon)k}{2\varepsilon - 1} \eqpun .
\end{align*}

To reason about the diverse proficiency of workers in a population,
we assume that worker effort functions are distributed according to a probability distribution $\mathcal{F}$.
We call a population of workers proficient if the following inequality holds
\begin{align}
        \label{theorem-ineq-popul-assumpt}
        \EE_{f\sim \mathcal{F}}[\sigma_f] \le \varepsilon \eqpun ,
\end{align}
where $\sigma_f$ is the solution of equation (\ref{eq-prof-level}) with a function $f$.

Note that a successful incentive scheme cannot guarantee that workers with limited proficiency have
an incentive to be truthful, because it might take too much effort for such workers to have error probability
smaller than $\varepsilon$.
%Workers in a supervision tree know their own effort functions, but they do not know effort function of other workers.
%In a supervision tree, workers play a game with incomplete information about other players type. 
The following theorem shows that all {\em interim} Bayes-Nash equilibria have the property that every proficient
worker plays with a truthful strategy.

\begin{theorem}
    \label{th-binarverif-hetero}
    Let workers be organized into a supervision tree with branching factor k
    and penalty $C> 0$. Let $\varepsilon \in (0, 1/2)$ and let $\mathcal{F}$ be a distribution of worker effort functions such that workers are proficient on average, i.e. inequality (\ref{theorem-ineq-popul-assumpt}) holds,
%    \begin{align}
%        \label{theorem-ineq-popul-assumpt}
%        \EE_{f\sim \mathcal{F}}[\sigma_f] \le \varepsilon \eqpun ,
%    \end{align}
    where $\sigma_f$ is the solution of equation (\ref{eq-prof-level}) with effort function $f$.
%    \begin{align}
%        \label{theorem-eq-sigma}
%        C = \frac{f(\sigma_f)k}{2\varepsilon - 1} \eqpun .
%    \end{align}
    If workers are rational, then all {interim} Bayes-Nash equilibria
    have the property that every proficient worker plays truthfully.
\end{theorem}
\begin{proof}
    In a supervision tree, the strategic choice of a player $u$ depends only on the players placed above $u$.
    We prove by induction of depth $l = 1, \dots, L-1$  of the tree that the best response of worker $u$ is to
    play with error probability $e^*_u$ such that $e^*_u < \sigma_u$, where $\sigma_u$ is the solution of
    equation (\ref{eq-prof-level}) with the worker's effort function $f_u$.
    Thus, by the definition of a proficient worker, it will follow that the best response of a proficient worker $u$
    is to play with  $e^*_u<\sigma_u \le \varepsilon$, i.e. to play truthfully.

    At level 1, a worker $u$ is evaluated by the supervisor who has error probability $e < \varepsilon$.
    Conditions of Lemma~\ref{lemma-binar} for worker $u$ are satisfied with $f=f_u$ and $\sigma = \sigma_u$.
    Thus, the best response of worker $u$ is to play with $e^*_u < \sigma_u$.

    To prove the inductive step,
    let us consider a worker $u$ at level $l+1$ and her superior
    $w$ at level $l$, with error probabilities $e_u$ and $e^*_w$ respectively.
    By the inductive assumption, $e^*_w < \sigma_w$.
    According to the first part of Lemma~\ref{lemma-binar}, the superior $w$ induces loss $L(e_u, e^*_w)$ to the worker $u$ that is defined by equation (\ref{eq-expected-bin-loss}).
    However, worker $u$ does not know the effort function of the superior.
    Thus, the loss of worker $u$ is computed as the expectation of $L(e_u, e^*_w)$ over different effort functions of the superior
    \begin{align*}
        \EE_{f_w \sim \mathcal{F}} L(e_u, e^*_w) &= \EE_{f_w \sim \mathcal{F}}[k f(e_u) + (1 - e_u)e^*_wC + e_u (1- e^*_w) C + e_u e^*_w D] \\
        &= k f(e_u) + (1 - e_u)\EE_{f_w \sim \mathcal{F}}[e^*_w]C + e_u (1- \EE_{f_w \sim \mathcal{F}}[e^*_w]) C + e_u \EE_{f_w \sim \mathcal{F}}[e^*_w] D \\
        &= L(e_u, \EE_{f_w \sim \mathcal{F}}[e^*_w]) \eqpun .
    \end{align*}
    We use $e^*$ to denote $\EE_{f_w \sim \mathcal{F}}[e^*_w]$.
    Minimizing the loss $\EE_{f_w \sim \mathcal{F}} L(e_u, e^*_w)$ with respect to $e_u$ is equivalent to
    minimizing the loss $L(e_u, e^*)$.
    Therefore, the strategic choice $e^*_u$ of worker $u$ in the presence of uncertainty about the type of the superior is equivalent to the strategic choice of worker $u$ with a superior that has error probability $e^*$.
    We show that $e^* < \varepsilon$ by using the inductive assumption $e^*_w < \sigma_w$ and 
    assumption (\ref{theorem-ineq-popul-assumpt}), that the workers are proficient on average:
    \begin{align*}
        e^* = \EE_{f_w \sim \mathcal{F}}[e^*_w] < \EE_{f_w \sim \mathcal{F}}[\sigma_w] \le \varepsilon \eqpun .
    \end{align*}
    The conditions of Lemma~\ref{lemma-binar} are satisfied for the worker $u$ with $f=f_u$, $\sigma=\sigma_u$, and $e_w=e^*$.
    It follows from Lemma~\ref{lemma-binar} that the best response of worker $u$ is $e^*_u < \sigma_u$.
    This finishes the inductive step and the proof of the theorem.
\end{proof}

\smallskip
\noindent{\bf What information do workers need?}
%\subsection{What information do workers need?}
The schemes considered in this section organize workers into hierarchies. 
What information do workers need to know about the hierarchy, as they set to do their work? 
Do they need to be given the precise hierarchical scheme, including the names (or identities) of their supervisors? 
Or can they just be told that a hierarchy exists, without being told even what their place in it is? 
The interest in these questions lies in the fact that revealing to workers the identity of those above and below them in the hierarchy could create incentives to communicate via secondary channels and sway the outcome. 

It turns out that the answer is somewhere in between: while workers do not need to know the identities of the workers above and below them in the hierarchy, they do need to know the level in which they are. 
The following pair of theorems makes this observation precise. 

We denote the pure defection strategy where a worker always reports the same solution for {\em any} tasks as $\xi$.

\begin{theorem}
    Assume workers are organized into a supervision tree but they are not told their level in the tree.
    Then, for each $\varepsilon > 0$, there are supervision trees where defecting with a constant strategy is a Nash equilibrium with a loss smaller than any truthful strategy.
\end{theorem}

\begin{proof}
    Let $N$ and $k$ be the number of players and the tree branching factor respectively.
    We analyze strategic choices of a worker $u\in U$ when all other workers $U\backslash u$ defect and play with strategy $\xi$.
    The worker can play a mix of the following two pure strategies.
    One strategy consists in playing the fixed move. 
    This strategy carries a cost when the worker picks the wrong outcome and is reviewed by the supervisor; this happens with probability $k / N$. 
    Thus, the expected cost of this strategy is bound by $k C / N$. 
    The other strategy consists in playing an outcome that differs from the constant being played by defectors. 
    Even leaving aside the cost of finding out the truth, this strategy carries a cost $(N - k) C / N$.
    So when $(N - k) C / N > k C / N$, or $N > 2k$, it is convenient to defect. 
    The result is intuitive: it is convenient to defect when the probability of being reviewed by another defector is larger than the probability of being reviewed by the single supervisor.
\end{proof}

The following theorem essentially says that telling workers their level in the hierarchy is the minimum and sufficient amount of information required to ensure that collaborating is the only Nash equilibrium. 

\begin{theorem}
If there is a fixed upper bound $k$ to the number of tasks that a worker is assigned, then the smallest amount of information a worker needs to know about the hierarchy 
to have an incentive to play with the truthfully strategy is $\Theta(
log \log N)$, where
$N$ is the number of players in the hierarchy, and $\Theta()$ is the big-Theta notation of complexity theory.
\end{theorem}

\begin{proof}
If we can give workers $\Theta(\log\log N)$ information or more, then we can tell them their level in the hierarchy, and the induction argument in Theorem~\ref{th-supervised-tree} applies.

Conversely, assume that we give fewer than $\Theta(\log\log N)$ bits of information to workers, and consider the situation for $N \rightarrow \infty$. 
The bits given out would induce a partition $C_1, C_2, \ldots, C_m$ of the workers, where workers receiving the same bits would belong to the same class.  
Assume that the partition classes are sorted according to size, so that 
$|C_1| < |C_2| < \cdots < |C_m|$.
As the number of bits is smaller than $\Theta(\log\log N)$, for every $\gamma > 0$, there are $n$ and $j$ so that $|C_j| < \gamma |C_{j+1}|$. 
In other words, as the number of classes is less than logarithmic in $N$, as $N$ grows, there must be arbitrarily large gaps in the ratios between sizes of adjacent classes. 
This implies that, for workers in $C_{j+1}$ as above, the probability of being reviewed by a worker in levels $C_1 \union \cdots \union C_j$ can become arbitrarily small, since those workers can check on at most 
$k^2 |C_1 \union \cdots \union C_j|$ workers below them. 
Thus, defecting becomes the preferred strategy by some of the workers if fewer than $\Theta(\log \log N)$ bits are communicated to the workers. 
\end{proof}

\iffalse
    \subsubsection{Can supervisors be lenient?}

    We consider a scenario when every superior in a supervision tree 
    evaluates at least two tasks by their subordinates, and we ask the question: can a superior forgive a worker when only one task is performed incorrectly, and only penalize when two or more tasks are performed incorrectly? 
    Unless a very high error probability is acceptable, the answer is negative. 
    In fact, if the superior forgives one task then a strategic worker will always deviate on one task.
    Indeed, when the worker diligently completes all but one of its tasks,
    the worker can reduce the cost by defecting on the last task, returning a random answer, without risk of penalty.
\fi

\subsection{The Quantitative Model}
    In the previous section we considered hierarchical schemes in the binary-verifiable setting.
    Workers report either a correct or incorrect answer in a task.
    We showed that hierarchical schemes provide incentives to be truthful in homogeneous worker populations if
    the penalty cost $C$ is large enough.
    We also extended the result to the case of heterogeneous worker populations where workers have different levels
    of proficiency.
    In this section we study a quantitative setting in which workers can give a real number as an answer to a task.
    We will show that hierarchical schemes provide incentives to be truthful and that the strength of the incentive does not deteriorate with the depth of the hierarchy.

    Let us consider a worker $u$ who assigns value $x$ to a task with answer $t\in \mathbb{R}$.
    Evaluation $x$ is a random variable with variance $\sigma^2=\EE(x-\EE[x])^2$ and bias $b=\EE[x] - t$.
    The expected error $v$ of the evaluation is $\EE(x-t)^2$.
    We show that the expected error $v$ can be represented as a sum of the variance $\sigma^2$ and squared bias
    $b^2$.
    Indeed, by adding and subtracting $\EE[x]$ within $\EE[(x-t)^2]$, expanding the squared sum and using the fact that $\EE[x - \EE[x]]= 0$, we obtain
    \begin{align*}
        \EE[(x - t)^2] &= \EE[(x - \EE[x] + \EE[x] - t)^2] = \EE[(x - \EE[x])^2] + (\EE[x] - t)^2  = \sigma^2 + b^2 \eqpun .
    \end{align*}

    The following proposition specifies the expected penalty of a worker $u$ with variance $\sigma^2_u $ and bias $b_u$ when evaluated by a superior $w$ with variance $\sigma^2_w$ and bias $b_w$.
    \begin{proposition}
        \label{prop-quant-penalty}
        Let workers $u, w\in U$ have variances $\sigma^2_u, \sigma^2_w$ and biases $b_u, b_w$ respectively.
        If worker $w$ supervises worker $u$ by assigning a penalty according to function (\ref{loss-quant}) with penalty constant $c>0$,
        then the expected penalty $l(\sigma_u, b_u, \sigma_w, b_w)$ of worker $u$ is 
        \begin{align}
            \label{expect-superv-penalty}
            l(\sigma_u, b_u, \sigma_w, b_w) = c(\sigma^2_u + b_u^2 - 2b_ub_w + \sigma^2_w + b_w^2) \eqpun .
        \end{align}
    \end{proposition}
    \begin{proof}
        We use $\EE_{u}$, $\EE_{w}$ to denote the expectations over the evaluation of workers $u$ and $w$ respectively.
        Let $x$ and $y$ be evaluations to a task by workers $u$ and $w$ respectively.
        We simplify the expected loss $\EE_u\EE_w[c(x-y)^2]$ by replacing expression $(x-y)^2$ with the equivalent expression
        $(x-t)^2 -2(x-t)(y-t) +(y-t)^2$, and using the independence of evaluations $x$ and $y$
        \begin{align*}
            l(\sigma_u, b_u, \sigma_w, b_w) &= c\left(\EE_u[(x - t)^2] - 2(\EE_u[x] - t)(\EE_w[y] - t) + \EE_w[(y - t)^2]\right) \\
            &= c(\sigma^2_u + b_u^2 - 2b_ub_w + \sigma^2_w + b_w^2) \eqpun .
        \end{align*}
    \end{proof}

    A worker $u$ has control over the expected estimation error $v$ that is the sum of $\sigma^{2}$ and $b^{2}$.
    To achieve the expected error $v$, worker $u$ makes effort $f_u(v)$, where $f_u$ is a strictly convex and decreasing function defined on $\mathbb{R}^+$.

    Let workers be organized into a supervision tree with branching factor $k$, and let
    worker $w$ be the parent of worker $u$.
    The expected cost of worker $u$ is the sum of two components: the cost of
    performing $k$ tasks, and the penalty due to the supervision by worker $w$.
    It directly follows from Proposition~\ref{prop-quant-penalty} that if superior $w$ is
    unbiased, i.e. $b_w = 0$, then the best response $v^*_u$ of worker $u$ is 
    \begin{align}
        \label{eq-bestresp-quant}
        v^*_u = \underset{v}{\arg\min}\, (kf_u(v) + cv) \eqpun .
    \end{align}
    Surprisingly, the best response of a worker to an unbiased superior does not depend on the precision of the worker's superior. This fact allows us to reason about the best response of a worker to an unbiased supervision without
    specifying the particular superior worker.

    We adopt the following natural assumption on the population of workers.
    We assume that the average bias of the best response to an unbiased supervision is $0$.
    This assumption does not restrict individual workers to be unbiased.

    Effort functions determine workers proficiency, as supervision penalty (\ref{loss-quant}) can provide
    incentives to be truthful only if the cost of performing a task does not outweigh the penalty.

    \begin{definition}
        Given $\varepsilon >0$ and penalty constant $c>0$ of loss (\ref{loss-quant}),
        a worker $u$ with effort function $f_u$ is {\em proficient} if the best response $v^*_u$
        %(\ref{eq-optresp-quant})
        to an unbiased supervisor is less than $\varepsilon$. 
    \end{definition}

    Because workers do not know each other's effort functions, strategic interaction of workers in a supervision
    tree can be formulated as a game with incomplete information.

    The following theorem shows that workers in a supervision tree have incentives to be truthful.
    
    \begin{theorem}
        \label{th-quant-hetero}
        Let rational workers be organized into a supervision tree with branching factor k
        and loss function (\ref{loss-quant}). 
        %Let the population of workers have the property that the average bias of workers under unbiased supervision is 0, then
        %And let the population of workers have the property that, when evaluated without a bias, the average best response worker bias is 0, then
        And let the population of workers have the property that the average bias of the best response to an unbiased supervision is 0, then
        all {interim} Bayes-Nash equilibria
        have the property that every proficient worker plays truthfully.
    \end{theorem}
    \begin{proof}
        We prove by induction of depth $l = 1, \dots, L-1$  of the tree that the best response of a worker $u$ is to
        play with $v^*_u$ that is the solution of optimization problem (\ref{eq-bestresp-quant}).
        Thus, by the definition of a proficient worker, it will follow that the best response of a proficient worker $u$
        is to be truthful.
        At level 1, a worker $u$ is evaluated by the supervisor who has bias 0, therefore
        worker $u$ plays with the expected error $v^*_u$.
        
        To prove the inductive step,
        let us consider a worker $u$ at level $l+1$ and their superior
        $w$ at level $l$.
        According to the  inductive assumption, worker $w$ plays with $v^*_w$ that is the best response to an unbiased supervision.
        Let $b^*_w$ be a bias of worker $w$.
        According to Proposition~\ref{expect-superv-penalty}, superior $w$ induces penalty $c(v_u - 2 b_u b^*_w + v^*_w)$ to worker $u$.
        Loss $L$ of worker $u$ under superior $w$ is 
        \begin{align*}
            L = kf_u(v_u) + cv_u - 2cb_u b^*w + v^*_w \eqpun .
        \end{align*}
        The expected loss $L$ across different types of superiors is
        \begin{align*}
            \EE[L] = kf_u(v_u) + cv_u - 2cb_u \EE[b^*_w] + \EE[v^*_w] \eqpun .
        \end{align*}
        Due to the assumption on the population of workers, we have $\EE[b^*_w] = 0$.
        Therefore, the best response of worker $u$ is to minimize $(kf_u(v_u) + cv_u )$, i.e. to play with
        the error probability $v^*_u$ (Equation \ref{eq-bestresp-quant}).
        This finishes the inductive step and the proof of the theorem.
    \end{proof}

\section{Incentives Schemes with Multiple Reviews per Item}
\label{sec-redundancy}

In the incentive schemes proposed in the previous sections,
many tasks will have only one worker assigned to it.
In this section, we consider the case of crowdsourcing with redundancy, i.e., when each tasks has multiple workers assigned to it. 
This can be useful when it is possible to aggregate the answers produced by the workers into a single, higher accuracy answer. 

\subsection{One Level Supervised Schemes}

When a task is performed by multiple workers, verifying a single task $i$ can be used to verify all the workers in $\partial i$. 
The supervisor can leverage this in order to try to minimize the number of tasks to be verified, while guaranteeing a worker verification probability $p$ that satisfies Theorem~\ref{th-lower-bound-p}.
We will show that when a graph $G$ of tasks and workers is given, i.e. when we do not have control over task allocation, then constructing the smallest subset $S$ is $\mathcal{NP}$-hard, and it can be proved by reduction from vertex-cover.
However, if we can control tasks allocation, then we can easily construct graphs on which the set of tasks that need verification is as small as possible.

%\subsubsection{The assignment graph is given.}
\smallskip
\noindent{\bf The assignment graph is given.}
We first study a scenario in which the worker-task assignment is fixed, and we must choose the subset $S\subseteq I$ of
tasks verified by the supervisor.
When the supervisor examines a task $i\in I$, she evaluates all the
workers $\partial i$ who were assigned to the task $i$.
Figure~\ref{fig-grading-all-workers} illustrates a case when examining 3 tasks is enough to evaluate all the workers.

\iffalse
    \begin{figure}[!ht]
    \centering
    \includegraphics[width=0.37\linewidth]{supervisor_assig.pdf}
    \caption{An example of a graph with all 6 workers being evaluated based on a set of 3 tasks.
    The supervisor inspects tasks 2, 4 and 5 that connected to all workers.
    The tasks and workers the supervisor reaches out are colored.}
    \label{fig-grading-all-workers}
    \end{figure}
\fi

The supervisor wants to spend the least amount of effort to evaluate at least $m$ workers.
For the case $p=1$, or $m = |U|$, the supervisor needs to find the smallest subset of tasks such that every worker is assigned one task from the subset.
We name the problem of finding such a set the {\it Superior Assignment} problem (abbreviated to {\it SA}).
The following theorem shows that the {\it Supervisor Assignment} Problem is $\mathcal{NP}$-hard.

\begin{theorem}
\label{th-vertex-cover}
Supervisor Assignment problem is $\mathcal{NP}$-hard.
\end{theorem}
\begin{proof}
We will show that finding the smallest vertex cover for any graph is an instance of the Supervisor Assignment problem.
Thus, solving the {\it Supervisor Assignment} problem is at least as hard as solving
Vertex Cover.

Let $G= (V, E)$ be an arbitrary graph with vertexes $V$ and edges $E$.
We construct a bipartite revision graph $G'$ for a set of workers $U$ and set of items $I$ by taking $U=E$ and $I=V$: that is, we use workers in our  
bipartite graph to represent the edges of the original graph.
Each worker $u \in U$ is assigned to review items $v_1, v_2$, where
in $G$ the edge $u$ connects $v_1$ and $v_2$.
The graph $G'$ is called the incidence graph \cite{IncidenceGraphBook}.
It is immediate to see that a subset of vertices $V' \subseteq V$ is a {\em vertex cover\/} for $G$ if and only if picking all items in $V'$ enables the verification of all workers $E$ of $G'$. 
Thus, Vertex Cover can be reduced to the Supervisor Assignment problem.
\end{proof}

We now show that, if every worker is assigned at most $k$ tasks, there are fast $k$-approximation algorithm for SA.
A $k$-approximation algorithm finds a subset $S'$ of tasks such that $|S'| < k |S|$, where $S$ is
the optimal solution.

We will show that the SA problem on graph $G$ is equivalent to the VC problem on a hypergraph with edge size at most $k$.
A hypergraph $H=(V, F)$ is a set of  vertices $V$ and hyperedges $F$.
A hyperedge $f\in F$ connects a subset of edges from $V$.
Hypergraph $H$ has edge size at most $k$ if every edge $f\in F$ contains at most $k$ nodes.
There are known simple $k$-approximation algorithms for VC on $k$-bounded hypergraphs \cite{halperin2002improved}.

\begin{proposition}
The Supervisor Assignment problem for a bipartite review graph $G = (U \union I, E)$ with degree at most $k$ is
equivalent to Vertex Cover for a hypergraph with edge size at most $k$.
\end{proposition}

\begin{proof}
The SA problem is immediately equivalent to a VC cover for a hypergraph that has $U$ as vertex set, and has $I$ as edge set, where each edge $i \in I$ connects the vertices that correspond to the workers to which $i$ is assigned.
\end{proof}

A simple $k$-approximation algorithm works as follows.
Let $G=(V\union F, E)$ be a bipartite graph and $S = \emptyset$.
While set $E$ is not empty,
we randomly choose an edge $(w, f) \in E$, add node $w$ to set $S$, and
delete all edges incident to $w$ or $f$.
When $E$ is empty, set $S$ is a $k$-approximation to the SA problem on graph $G$.

%\subsubsection{The assignment graph can be constructed.}
\smallskip
\noindent{\bf The assignment graph can be constructed.}
\iffalse
    \begin{figure}[!ht]
    \centering
    \includegraphics[width=0.5\linewidth]{pathological_graph.pdf}
    \caption{An example of review graph where the supervisor evaluates all workers by reviewing only one task.}
    \label{fig-pathological-graph}
    \end{figure}
\fi
If we can construct the assignment graph, then it is easy to ensure optimality.
%One extreme example is depicted on figure~\ref{fig-pathological-graph}.
When all workers have only one task in common,
the supervisor can evaluate all the workers by verifying only one task.
From a crowdsourcing perspective, however, concentrating effort of all workers on
one assignment has the unwelcome effect that all other tasks receive fewer
workers.
If we use worker multiplicity for a task in order to achieve higher reliability in the solution of a task, this is undesirable.
A natural assumption is to require the review graph $G$ to be $k$-regular.

To construct a $k$-regular review graph, we proceed as follows.
We select $n=\lceil \frac{|U|}{k} \rceil$ ``peg''  tasks first.
Each of these peg tasks will be done by a set of $k$ non-overlapping workers, so by verifying the $n$ peg tasks, the supervisor is able to verify all workers ($p=1$). 
For smaller values of the verification probability $p$, the supervisor can simply choose to verify a randomly chosen subset of the peg tasks. 
We assume, of course, that the workers cannot compare their work with each other, so that they cannot infer which tasks are the peg tasks among those they are assigned.
Once the peg tasks and their reviewers are chosen, we assign the other tasks to workers in any way that leads to $k$-regularity.
It is easy to see that this construction is optimal, for $|U|$ workers doing $k$ tasks each cannot be verified by picking fewer than $n$ items. 
\iffalse
    Figure~\ref{fig-pegs} illustrates the construction for $k=3$ and $|U|=6$.

    \begin{figure}[!ht]
    \centering
    \includegraphics[width=0.5\linewidth]{pegs.pdf}
    \caption{An example of a 3-regular graph where 2 peg tasks (in red) connect all workers.}
    \label{fig-pegs}
    \end{figure}
\fi

\subsection{Hierarchical Supervised Schemes}
\label{subsec-general-hierarch}

A {\it supervision hierarchy} combines a bipartite graph of workers and tasks
and a supervision tree.
The supervision tree provides an incentive while the bipartite graph ensures that every task
is assigned to several workers.

\begin{definition}
\label{def-sup-heirarch}
A supervision hierarchy is a connected graph that consists of two subgraphs: a bipartite graph
$G=(U\union I, E)$ and a supervision tree $T$ with workers $U_T$ and tasks $I_T$.
The set of tasks $I_T$ is a subset of tasks $I$ and for every worker $u\in U$
there is a task $i\in I_T$ such that the edge $(u, i)$ belongs to $E$.
\end{definition}

Figure~\ref{fig-general-hierarchy} illustrates such a supervision hierarchy.
The supervisor provides an incentive for the two immediate subordinate workers while 
these workers provide the incentive to the rest of workers by performing a total of 4 tasks.

\begin{figure}
  \centering
  \begin{subfigure}[b]{0.4\textwidth}
    \centering
    \includegraphics[width=0.87\textwidth]{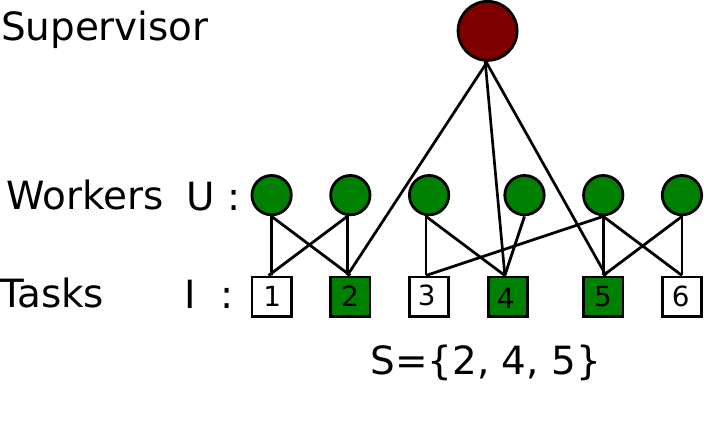}
    \caption{An example of a graph with all 6 workers being evaluated based on a set of 3 tasks.
The supervisor inspects tasks 2, 4 and 5 that connected to all workers.
The tasks and workers the supervisor reaches out are colored.}
    \label{fig-grading-all-workers}
  \end{subfigure}\hspace{.03\textwidth}%\hfill
  \begin{subfigure}[b]{0.50\textwidth}
    \centering
    \includegraphics[width=0.8\textwidth]{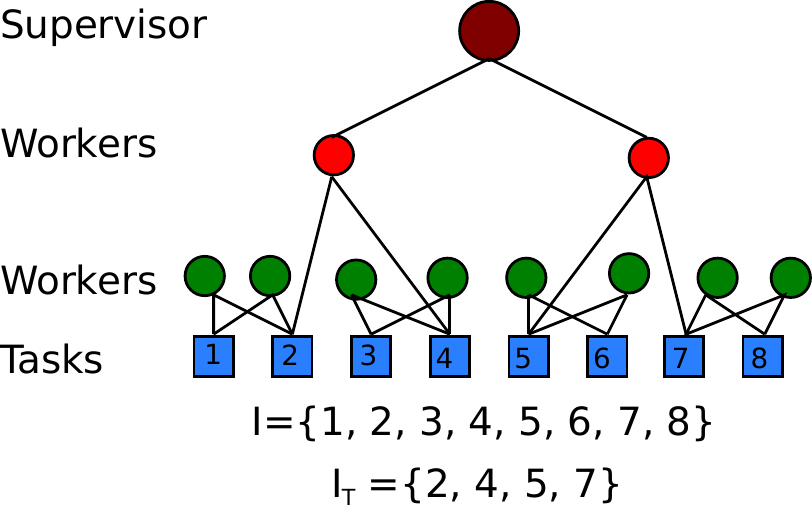}
    \caption{A supervision hierarchy that is a union of a supervision tree and
a bipartite graph of workers and tasks.
Every task is assigned to at least 2 workers.
The set of tasks $I_T$ in the tree is a subset of tasks $I$ in the bipartite graph.
Every worker is assigned at least one tasks from the set $I_T$.}
    \label{fig-general-hierarchy}
  \end{subfigure}
  \caption{A bipartite graph evaluated by the supervisor, and a supervision hierarchy.}\label{fig:animals}
\end{figure}

\iffalse
    \begin{figure}[!ht]
    \centering
    \includegraphics[width=0.4\linewidth]{general_hierarchy.pdf}
    \caption{A supervision hierarchy that is a union of a supervision tree and
    a bipartite graph of workers and tasks.
    Every task is assigned to at least 2 workers.
    The set of tasks $I_T$ in the tree is a subset of tasks $I$ in the bipartite graph.
    Every worker is assigned at least one tasks from the set $I_T$.
    }
    \label{fig-general-hierarchy}
    \end{figure}
\fi

For a given bipartite graph $G$ the task of constructing the smallest supervision
hierarchy is $\mathcal{NP}$-hard.
Indeed, the subset $I_T$ of $I$ has the property that every worker $u\in U$ has at least
one task from $I_T$.
Thus finding the smallest set $I_T$ is an instance of the {\it Supervision Assignment} problem we discussed in the previous section;
and showed that it is an $\mathcal{NP}$-hard problem.

Theorems~\ref{th-binarverif-hetero}, \ref{th-quant-hetero} can be extended to show that
supervision hierarchies provide incentives to be truthful for the binary-verifiable and the quantitative settings respectively.
Indeed, the theorems hold directly for workers that belong to the tree $T$ of the supervision hierarchy.
On the other hand, every worker in graph $G$ of the supervision hierarchy has a superior worker from the bottom level of tree $T$.
Thus, workers in $G$ can be considered as one extra layer of workers within tree $T$.

\section{Conclusions}
We proposed and analyzed supervision incentive schemes that ensure that the optimal strategy for workers is to be truthful. 
The schemes rely on hierarchies in order to scale to arbitrarily large sets of items and workers, while requiring only a constant amount of work on the part of the supervisor. 
In the hierarchy, workers are organized in layers, and every layer exerts an incentive over the layer below, ensuring that the optimal behavior of the workers is sufficiently precise. 
We show that the truthful incentive holds even in populations of workers with diverse proficiency, where workers can have limited proficiency, provided that there are enough proficient workers in the crowd.
Interestingly, the only information the workers need to know about the hierarchy is their level in it: they do not need to know the identities of their supervisors or subordinates, nor which tasks they share, and all workers perform exactly the same work. 
In particular, there are not two flavors of ``normal'' and ``metareview'' tasks. 
Our schemes graciously extend from binary verifiable tasks to quantitative
tasks making them relevant to a wide range of crowdsourcing applications.

%\bibliographystyle{abbrv}

%\bibliographystyle{ACM-Reference-Format}
%\bibliography{references}

\appendix
\section[Proof of Lemma]{Proof of Lemma~\ref{lemma-binar}}
\label{sec-lemma-proof}
\iffalse
    \begin{lemma}
        \label{lemma-binar}
        Let workers $u, w \in U$ have error probabilities $e_u, e_w$ respectively,
        and let worker $w$ be the parent of worker $u$ in a supervision tree
        with branching factor $k$ and penalty cost $C>0$.
        If worker $u$ has effort function $f$ then the expected loss $L(e_u, e_w)$
        of worker $u$ under the supervision of worker $w$ is 
        \begin{align}
            \label{eq-expected-bin-loss}
            L(e_u, e_w) = kf(e_u) + e_u(1-e_w)C + (1-e_u)e_wC + e_ue_wD \eqpun ,
        \end{align}
        where $D$ is a constant from the $[0, C]$ interval.
        Moreover, if there is $\sigma, \varepsilon \in (0, 1/2)$ such that
        $e_w < \varepsilon$ and 
        \begin{align}
            \label{lemma-ineq-C}
            C \ge \frac{f'(\sigma)k}{2\varepsilon - 1} \eqpun ,
        \end{align}
        then every $e^*_u\in \underset{e_u}{\arg\min} L(e_u, e_w)$ satisfies inequality
        $e^*_u < \sigma$ .
      %  \begin{align*}
      %      e^*_u < \sigma \eqpun .
      %  \end{align*}
    \end{lemma}
\fi
\begin{proof}
    The expected loss $L(e_u, e_w)$ of worker $u$ consists of 4 components:
    \begin{align*}
        L(e_u, e_w) = k f(e_u) + (1 - e_u)e_wC + e_u (1- e_w) C + e_u e_w D \eqpun .
    \end{align*}
    The first component $kf(e_u)$ is due to the effort of performing $k$ tasks with error probability $e_u$.
    The other components account for the penalty that the superior $w$ imposes in three mutually exclusive events.
    In particular, the second component $(1 - e_u)e_w C$ accounts for the event where
    worker $u$ provides the correct solution to the common task but 
    the superior makes a mistake.
    The third component $e_u (1 - e_w)C$ accounts for the case where
    the worker makes a mistake and the superior is correct.
    The fourth component $e_u e_w D$ accounts for the event where both worker and superior are incorrect; the penalty $D$ belongs to the interval $[0, C]$, depending  on the probability of the event when the worker and the superior have different answers and both of them are incorrect.

    The loss $L(e_u, e_w)$ is a convex function of $e_u$ as a combination of convex and linear functions.
    Therefore, the set $\arg\min L(e_u, e_w)$ is not empty.
    Let $e^*_u\in \arg\min L(e_u, e_w)$.
    Error probability $e^*_u$ satisfies the following inequality
    \begin{align*}
        k f'(e^*_u) - e_wC + (1 - e_w)C + e_wD = 0 \Rightarrow
        f'(e^*_u) = \frac{(2e_w - 1)C - e_wD}{k} \eqpun .
    \end{align*}
    Combining it  with the fact that $e_w, D \ge 0$, 
    and with the assumption that $e_w < \varepsilon$, we obtain
    \begin{align}
        \label{lemma-ineq-hierarch-proof-1}
        f'(e^*_u) &\le \frac{(2e_w - 1)C}{k} < \frac{(2\varepsilon - 1)C}{k}  \eqpun .
    \end{align}
    Note that the $(2\varepsilon - 1)$ multiplier is negative as $\varepsilon < 1/2$.
    Because $(2\varepsilon - 1) < 0$, the right hand side of inequality (\ref{lemma-ineq-hierarch-proof-1}) can be bounded using inequality (\ref{lemma-ineq-C})
    \begin{align}
        \label{lemma-ineq-hierarch-proof-2}
        \frac{(2\varepsilon - 1)C}{k} \le \frac{(2\varepsilon - 1)}{k}\frac{f'(\sigma) k}{2\varepsilon - 1} = f'(\sigma) \eqpun .
    \end{align}
    From inequalities (\ref{lemma-ineq-hierarch-proof-1}, \ref{lemma-ineq-hierarch-proof-2}), it follows that $f'(e^*_u) <  f'(\sigma)$.
    Because function $f'$ is increasing as the derivative of a strictly convex
    function \cite{rockafellar},
    we conclude that $e^*_u < \sigma$.
\end{proof}

\section[Proof of Theorem]{Proof of Theorem~\ref{th-counter-example}}%\refcounterexample}
\label{th-counter-example-proof}

To prove the theorem, we will show that when 
 the cost $C$ satisfies inequality
\begin{align}
    \label{ineq-th-counter-example}
    C < \frac{f'(\varepsilon)k}{2\varepsilon - 1} \eqpun ,
\end{align}
then there exists an effort function, a supervision tree, and a level $l$ such that a rational worker chooses a strategy with error probability greater than $\varepsilon$.
In particular, we choose effort function to be $f(x) = -\ln(x)$, and we
assume that (\ref{ineq-th-counter-example}) holds for $\varepsilon \in (0, 1/4)$.
To construct a tree we assume that the solution set $A$ consists of 2 elements only.
Let $e_t$ be error probability of a worker on depth $t$ of the supervision tree.
We need to show that for large enough $l$ a rational worker chooses $e_l>\varepsilon$.
The superior on level 0 provides correct solutions, thus $e_t = 0$.
The expected loss $L_t(e_t)$ of a worker $u$ on level $t$ is
\begin{align}
    \label{eq-th-counter-example-1}
    L_t(e_t) = kf(e_t) + e_t (1-e_{t-1})C + (1 - e_t) e_{t-1}C \eqpun .
\end{align}
Worker $u$ minimizes their loss by choosing $e_t$ that sets the derivative of the loss to 0
\begin{align*}
    kf'(e_t) + (1-2e_{t-1})C = 0 \Rightarrow f'(e_t) = \frac{(2e_{t-1} - 1)C}{k} \eqpun .
\end{align*}
Given that $f'(x) = -1/x$, the optimal value $e_t$ is $\frac{k}{(1 - 2e_{t-1})C}$.
%\begin{align*}
%    e_t = \frac{k}{(1 - 2e_{t-1})C}
%\end{align*}
The difference between $e_t$ and $e_{t-1}$ is
\begin{align}
    \label{eq-th-counter-example-2}
    e_t - e_{t-1} = \frac{k/C}{1-2e_{t-1}} - e_{t-1} = \frac{k/C + 2e_{t-1}^2 - e_{t-1}}{1-e_{t-1}} \eqpun .
\end{align}

We are going to find a constant $\Delta > 0$ such that $e_t - e_{t-1} >  \Delta$
for any $t\ge0$.
This would mean that as the tree depth increases, the probability of errors by workers would steadily increase, eventually surpassing the truthfulness threshold $\varepsilon$.
To bound the right hand side of (\ref{eq-th-counter-example-2}), we note that $1/(1-2e_{t-1}) \le 1$ for $e_{t-1} \ge 0$.
Therefore
\begin{align}
    \label{ineq-the-counter-example-4}
    e_t - e_{t-1} \ge k/C + 2e_{t-1}^2 - e_{t-1} \eqpun .
\end{align}
Note that function $g(e_{t-1}) = 2e_{t-1}^2 - e_{t-1}$ is monotonically decreasing on the interval $[0, \varepsilon]$ as $\varepsilon < 1/4$.
Assume that that all levels workers play with a truthful strategy, that is,  $e_t < \varepsilon$ for any $t\ge 0$.
We can further bound $e_t - e_{t-1}$ by using inequality $2e_{t-1}^2 - e_{t-1} > 2\varepsilon^2 - \varepsilon$ and inequality (\ref{ineq-the-counter-example-4})
\begin{align}
    \label{ineq-th-counter-example-2}
    e_t - e_{t-1} > k/C + 2\varepsilon^2 - \varepsilon \eqpun .
\end{align}

Inequality (\ref{ineq-th-counter-example}) implies that for some $\delta >0$,
$C=\frac{kf'(\varepsilon)}{2\varepsilon - 1} - \delta$.
We use it to simplify the right hand side of (\ref{ineq-th-counter-example-2})
\begin{align*}
    e_t - e_{t-1} > k/C + 2\varepsilon^2 - \varepsilon = \frac{k}{\frac{k}{\varepsilon(1-2\varepsilon)} - \delta} -\varepsilon(1 - 2\varepsilon) \eqpun .
\end{align*}
For brevity, we denote $\varepsilon(1-2\varepsilon)$ as $a$.
\begin{align*}
    e_t - e_{t-1} > \frac{k}{\frac{k}{a} - \delta} - a = \frac{ak}{k-a\delta} - a =  \frac{a^2\delta}{k - a\delta} \eqpun .
\end{align*}
Expression $\frac{a^2\delta}{k - a\delta}$ is greater than 0 and does not depend on the level $t$ of the hierarchy; we denote it as $\Delta$.
\begin{align}
    \label{ineq-th-counter-example-3}
    e_t - e_{t-1} > \Delta \eqpun.
\end{align}
The derivation of (\ref{ineq-th-counter-example-3}) is based on the assumption that $e_t < \varepsilon$.
If we choose a hierarchy level $l$ such that
$l > \frac{\varepsilon}{\Delta}$, it follows from inequality (\ref{ineq-th-counter-example-3}) that $e_l - e_0 > \Delta * l > \varepsilon$.
Because $e_0 = 0$, we conclude that $e_l > \varepsilon$ which contradicts to our
assumption that $e_t <\varepsilon$ for $t \ge 0$.
We have shown there exists a hierarchy level $l$ such that $e_l > \varepsilon$ and the
supervision tree does not provide incentive past depth $l$.

\end{document}